\documentclass[12pt,a4paper,normalheadings,headsepline,headinclude,bibtotoc,fleqn,review]{article}
\usepackage[latin1]{inputenc}
\usepackage{amsmath}
\usepackage{booktabs}
\usepackage{array}
\usepackage{amsthm}
\usepackage{dcolumn}
\usepackage{endnotes}
\usepackage{units}
\usepackage{marvosym}
\usepackage[lines=10]{geometry}
\usepackage{longtable}
\usepackage{rotating}
\usepackage{graphicx}
\usepackage{wrapfig}
\usepackage{expdlist}
\usepackage{amssymb}
\usepackage{gloss}
\usepackage{nomencl}
\usepackage{natbib}
\newtheorem{rem}{Remark}

\newtheorem{prop}{Proposition}

    \makegloss

\renewcommand{\thesection}{\arabic{section}}

\usepackage{fancyhdr} 
\begin{document}

\begin{center}\Large{Endogenous Fertility Waves and the Dynamics of Utility in an Overlapping Generations Model}\footnote{Much of this paper was written during my time at the Max Planck Institute in Bonn as well as during my time at Zhejiang University.}\end{center}
\begin{center} \textit{Wolfgang Kuhle}\\ \textit{Corvinus University Budapest, Budapest, Hungary, E-mail wkuhle@gmx.de\\ MEA, Max Planck Institute for Social Law and Social Policy, Munich, Germany}\end{center}

\nocite{wei71}

\noindent\emph{\textbf{Abstract:} This paper investigates the conditions under which the Easterlin hypothesis holds within a neoclassical overlapping generations model with endogenous capital accumulation, wages, interest rates, and fertility. We develop a tractable analytical framework that maps economic transitions into utility space via a continuously differentiable first-order difference equation for cohort lifetime utilities. This reformulation allows for a transparent normative evaluation of non-steady-state paths without requiring explicit solutions to the underlying nonlinear system. Within this framework, we show that when fertility cycles emerge and children are normal goods, the utility of small cohorts strictly exceeds that of large cohorts. Crucially, this cohort-welfare asymmetry is driven by fertility preferences and is independent of the economy's position relative to the golden rule.}


\section{Introduction}\label{sg1}

The \citet{Eas61} hypothesis posits a self-generating mechanism of fertility waves: a decline in labor market entry raises fertility, which, with a two-decade lag, increases labor market entry and subsequently lowers fertility. While the literature has mainly studied either (i) endogenous fertility waves without capital accumulation as in \citet{Sam766} or (ii) steady state models with endogenous capital accumulation and exogenous fertility,\footnote{\citet{Sam75}, \citet{Mic93}, \citet{Jae09}, \citet{Ponth2010}, \citet{cro12}, and \citet{Fel16}.} the present paper merges these strands. That is, we focus on non-steady state paths in which capital accumulation, wages, interest rates, and fertility are endogenous. Within this framework, we examine the conditions under which the model generates fertility cycles consistent with the Easterlin hypothesis. Moreover, we show that whenever the model generates Easterlin-type fertility cycles, the utility of individuals born into small cohorts strictly exceeds the utility of those born into large cohorts. 

To evaluate the welfare consequences of such fluctuations, we introduce an approach based on the law of motion of lifetime utility. That is, we first show that the standard neoclassical OLG model implies that cohort utilities evolve according to a first-order difference equation, \( U_{t+1} = \xi(U_t) \). The advantage of analyzing dynamics in ``utility space" rather than ``capital space" (\( k_{t+1} = \psi(k_t) \)) is that it directly records the lifetime utilities of successive cohorts, thereby providing a normative evaluation of different non-steady state growth paths. Moreover, the approach applies to general utility and production functions and does not require explicit solutions of the highly nonlinear system. The present analytical technique thus complements numeric approaches, e.g., by  \citet{Lud07b}, \citet{LudVog10} and \citet{LudSchVog12}.  

In Section \ref{sg2} we derive the dynamics of utility for the textbook OLG model with exogenous fertility, showing that utility and capital converge at the same rate, but not necessarily in the same direction. That is, utility may fall while the capital labor ratio is increasing, and vice versa. 

In Section \ref{sect32} we introduce endogenous fertility and use the dynamics of utility to characterize the correlation between cohort size and lifetime utility. In particular, we show that small cohorts enjoy a lifetime utility greater than that of large cohorts, provided that children are normal goods. Moreover, we show that this condition is independent of the golden rule criterion, which is pivotal in steady state utility comparisons.

The paper is organized as follows. Section \ref{sg2} presents the dynamics of utility with exogenous fertility. Section \ref{sect32} discusses the endogenous fertility case. Section \ref{rel} relates our findings to the literature. Section \ref{sg4} concludes.

\section{The Dynamics of Utility}\label{sg2}

Production is characterized by capital and labor inputs, $K_t$ and
$L_t$, which produce aggregate output $Y_t=F(K_t,L_t)$, where
$F(,)$ is concave, first-degree-homogenous, and satisfies the
Inada conditions. Output per unit of labor in period $t$ can thus
be written as:
\begin{eqnarray} y_t=f(k_t), \quad k_t:=\frac{K_t}{L_t}, \quad f'(k_t)>0,\quad f''(k_t)<0.   \label{1}\end{eqnarray}
Factor markets are competitive such that interest $r_t$ and wages
$w_t$ are
\begin{eqnarray} r_t=f'(k_t),\quad w_t=f(k_t)-f'(k_t)k_t.\label{2}\end{eqnarray}

Population grows at rate $n$, and each cohort lives for two
periods. In the first period of life, each individual supplies one
unit of labor inelastically. The second period of life is spent in
retirement. The associated life-cycle savings problem for an
individual born in period $t$ is thus given by:
\begin{eqnarray}
U_t=\max_{s_t,c^1_t,c^2_{t+1}}U(c_t^1,c_{t+1}^2) \quad s.t. \quad
c^1_t=w_t-s_t, \quad
c^2_{t+1}=s_t(1+r_{t+1}),\label{3}\end{eqnarray} such that the
first-order condition
\begin{eqnarray}\frac{U_{c^1}}{U_{c^2}}=1+r_{t+1}\quad
\Leftrightarrow\quad s_t=s(w_t,r_{t+1}), \quad 0<s_w<1,\quad
s_r\gtreqqless0, \label{4}\end{eqnarray} implies savings.
Regarding utility (\ref{3}), we assume that it is
concave, twice continuously differentiable, and that it satisfies
the Inada conditions, such that savings in (\ref{4}) are positive, unique,
and continuously differentiable in its two arguments $w_t$ and
$r_{t+1}$. Moreover, we assume that first- and second-period
consumption are both normal goods, such that the propensity to save
out of income satisfies $0<s_w<1$.

Finally, the life-cycle savings condition
\begin{eqnarray}  (1+n)k_{t+1}=s(w_t,r_{t+1}) \label{5} \end{eqnarray}
closes the model. To obtain the law of motion of the capital
stock, we substitute (\ref{2}) into (\ref{5}) to obtain:
\begin{eqnarray} (1+n)k_{t+1}=s(f(k_t)-f'(k_t)k_t,f'(k_{t+1})).  \label{6}\end{eqnarray}
Regarding (\ref{6}), we follow \citet{Gal89} and \citet{Cro02} and
assume the standard conditions for the existence of a unique
monotonically stable steady state equilibrium, where
$k_{t+1}=k_t=k$ and
\begin{eqnarray} &&0<\frac{dk_{t+1}}{dk_t}=\frac{-s_wf''(k)k}{1+n-s_rf''(k)}<1,\label{7}\end{eqnarray}
are satisfied. To study the dynamics of utility, we now rewrite the model (\ref{1})-(\ref{7}) in terms of utilities
$U_t$ rather than capital intensities $k_t$. That is, rather than tracking out of steady state transitions via capital intensities, our goal is to evaluate macroeconomic adjustments directly through the lens of lifetime welfare. To change the system's characterization from traditional "capital space" to "utility space," we must establish that each capital stock $k_t$ maps uniquely to a specific lifetime utility level $U_t$. Proposition \ref{c1} formalizes this coordinate transformation:

\begin{prop}\label{c1} Let $U$ be the steady state level of life-cycle utility.
For a neighborhood $\mathcal{N}$ of $U$, there exists a
continuously differentiable first-order difference equation such
that $U_{t+1}=\xi(U_t)$, $\xi\in\mathcal{C}^1$, if life-cycle
utility $U_t$ is in this neighborhood. Moreover, in this
neighborhood of $U$, we have
$\frac{dU_{t+1}}{dU_t}=\xi_u(U)=\psi_k(k)=\frac{dk_{t+1}}{dk_t}$, such that
capital and utility converge at the same rate to the steady
state.\end{prop}
\begin{proof} By our assumptions, $U(,),s(,),f'(),w(),r()$ are all at least once continuously
differentiable ($\mathcal{C}^1$). Moreover, except for
$s_r$, the functions' partial derivatives are all non-vanishing.
Therefore, the implicit
function theorem ensures that there exist functions $\phi()$ and $\psi()$,
such that:\footnote{In Appendix \ref{a0}, we show that the
technical conditions that allow us to apply the implicit function
theorem to derive (\ref{8}), (\ref{9}), and (\ref{10}) are indeed
satisfied. With the formalism in place, we prove a stronger result
than Proposition \ref{c1}. Namely, we show that $\xi()$ is a
$\mathcal{C}^1$ diffeomorphism. That is, $\xi()$ and its inverse
$\xi^{-1}()$ are both continuously differentiable. Hence, we can
also write $U_{t-1}=\xi^{-1}(U_{t})$ and use current utilities
$U_{t}$ to calculate past utilities $U_{t-1}$.}
\begin{eqnarray} &&(1+n)k_{t+1}=s(w(k_t),r(k_{t+1})) \quad \Leftrightarrow \quad k_{t+1}=\psi(k_t), \quad \psi\in
\mathcal{C}^1\label{8}\end{eqnarray} and
\begin{eqnarray}U_t&&=U\Big(w(k_t)-s(w(k_t),r(k_{t+1})),(1+r(k_{t+1}))s(w(k_t),r(k_{t+1}))\Big)=\tilde{U}(k_t,k_{t+1})\nonumber\\
&&=_{|(\ref{8})}\tilde{U}(k_t,\psi(k_t)) \quad \Leftrightarrow
\quad k_t=\phi(U_t), \quad \phi\in
\mathcal{C}^1\label{9}\end{eqnarray} Substituting (\ref{9}) into
(\ref{8}) yields a first-order difference equation in
$U_t$:\footnote{By the implicit function theorem, we know
(\citet{Lee03}, pp. 164-166) that $\phi^{-1}\in\mathcal{C}^1$
exists. Moreover, we know that $\xi\in\mathcal{C}^1$, since the
composition of continuously differentiable functions is itself
continuously differentiable. See Appendix \ref{a0} for details on
$\phi^{-1}()$.}
\begin{eqnarray} \phi(U_{t+1})=\psi(\phi(U_t)) \quad \Leftrightarrow \quad U_{t+1}=\phi^{-1}(\psi(\phi(U_t)))=:\xi(U_t),\quad \xi\in\mathcal{C}^1   \label{10} \end{eqnarray}

With this continuously differentiable first-order difference
equation in place, it remains to show that utility and capital
converge at the same rate. Once we differentiate (\ref{10}) around
the steady state, where $U_{t+1}=U_t=U$, the condition for
monotonous convergence is:
\begin{eqnarray} 0<\frac{dU_{t+1}}{dU_t}=\frac{\phi_u(U)\psi_k(\phi(U))}{\phi_u(U)}=\psi_k(\phi(U))<1   \label{11} \end{eqnarray}
To relate condition (\ref{11}) to the convergence of the capital
stock, we recall that
$\psi_k(\phi(U))=_{|(\ref{9})}\psi_k(k)=\frac{dk_{t+1}}{dk_t}$ and
note that (\ref{7}) implies convergence, $\psi_k(k)\in(0,1)$.

Finally, we note that differentiation of (\ref{10}) yields $\frac{dU_{t+1}}{dU_t}=\xi_u(U)=_{|(\ref{11})}\psi_k(k)$. Hence, utility and capital converge monotonically and at the same rate $\psi_k(k)$. 
\end{proof}

\begin{center}
\begin{figure}
\begin{center}
\scalebox{1}{\includegraphics*{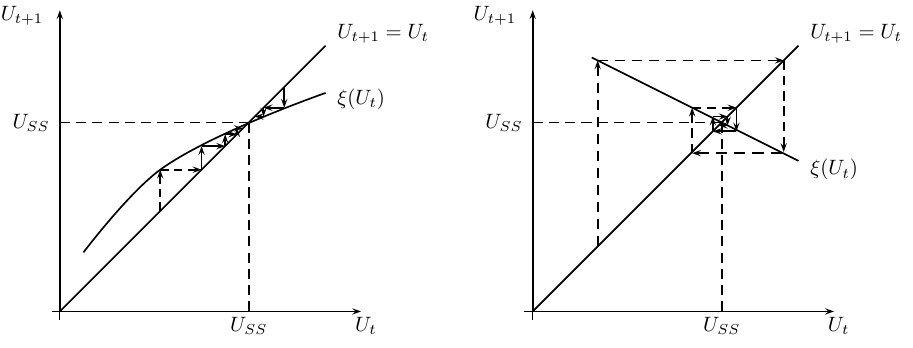}}
\caption{\emph{The Dynamics
of Utility.}}\label{g1}\end{center}\emph{}
\end{figure}
\end{center}

\begin{rem} It is straightforward to show that Proposition \ref{c1} extends to the case where capital and thus utility converge in cycles as
depicted on the right in Diagram $\ref{g1}$. That is, where
$\xi_u=\psi_k\in(-1,0)$.\end{rem} 

\begin{rem} Proposition \ref{c1} applies to transition paths near the steady state. If the economy is far from the final steady state, we have
$\frac{dU_{t+1}}{dU_t}=\frac{\phi_u(U_t)\psi_k(\phi(U_{t}))}{\phi_u(U_{t+1})}\neq\psi_k(\phi(U_{t}))$,
and utility and capital do not converge at the same rate along the
transition path. Related, the derivatives $\phi_u$ and
$\phi^{-1}_k$ can change signs along the transition path. That is,
initially utility can be monotonically increasing at first but later
decreasing, if the economy moves from efficient to a sufficiently
inefficient capital stock. \end{rem}

\subsection{Capital ``in'' and ``near'' the Steady State}\label{sect31}
Proposition \ref{c1} shows that utility and capital converge at
the same rate $\psi_k$ to the steady state. This, however, does
not mean that utility and capital move in the same direction; nor
does it mean that utility near the steady state decreases in
capital if steady state utility does. More precisely, we have
\begin{prop} Utility near the steady state increases
with the capital stock, such that $sign\{dU_t\}=sign\{dk_t\}$, if
and only if $\frac{1+r}{1+n}>\psi_k$.\label{c2}\end{prop}
\begin{proof} Using (\ref{9}), we write
\begin{eqnarray} &&U_1=\xi(U_0)\nonumber\\
&&U_2=\xi(U_1)=\xi(\xi(U_0))   \nonumber\\
&&...\nonumber\\
&& U_t=\xi^{(t)}(U_0) \label{100}\end{eqnarray} Applying the
chain rule to (\ref{100}) yields the change in cohort $t's$
utility $\frac{dU_t}{dU_0}=\xi_u^t=_{|(\ref{11})}\psi_k^t$
associated with a change $dU_0$. That is, $dU_0>0$ ($dU_0<0$)
implies $dU_t>0$ ($dU_t<0$) as long as $0<\psi_k$. Regarding capital we recall the
definition of $\phi$ in equations (\ref{92})-(\ref{94}), i.e.,
that
$\phi^{-1}_k=\frac{dU_t}{dk_t}=-U_{c^1}\Big((1+r)-\psi_k(1+n)\Big)\frac{1}{1+r}f''(k)k$.
capital and utility move in the same direction as long as
$(1+r)-\psi_k(1+n)>0$. Hence, capital and utility move in the same direction as long as
$(1+r)-\psi_k(1+n)>0$.\end{proof}

Proposition \ref{c2} shows that the relation between capital and
utility can be very different along the transition path than in the steady state, where $k_{t+1}=k_t=k$ and
thus $\frac{dU}{dk}=-U_{c^1}(r-n)f''(k)k$. This difference thus hinges
on the speed of convergence $\psi_k$. Along a transition path,
increments in capital $dk_t>0$ increase wages by
$dw_t=-f''kdk_t>0$ and reduce interest
$dr_{t+1}=\frac{dk_{t+1}}{dk_t}f''dk_t=\psi_kf''dk_t<0$. This
reduction in interest, however, is smaller than the reduction in
steady state interest, which amounts to $dr=f''dk$. Hence,
increases in capital can increase utility even if the economy
grows on a path where capital exceeds the golden rule level.

This difference is more pronounced in economies that converge
rapidly to the steady state, i.e., where $\psi_k$ is small. Rapid
convergence to the steady state increases the differences between
changes in utility \emph{near} the steady state and changes in
utility \emph{in} the steady state. In contrast, economies
that converge slowly to the steady state behave more similar to
economies that are in a steady state. 

\section{Endogenous Fertility Cycles}\label{sect32}

Compared to the baseline model, members of cohort $t$ now choose
savings and fertility $n_{t}$ when they are working aged:
\begin{eqnarray}
&&U_t=\max_{s_t,c^1_t,c^2_{t+1},n_{t}}U(c_t^1,c_{t+1}^2,n_{t}),\quad U_n>0, \quad U_{nn}<0\label{171}\\
&& \quad s.t. \quad c^1_t=w_t-s_t-c(n_{t}), \quad c'()>0, c''()>0,
\quad c^2_{t+1}=s_t(1+r_{t+1})\nonumber\end{eqnarray} Where utility (\ref{171}) assumes that agents trade off the cost $c(n)$ of having children with the pleasures of parenthood. As before, we assume that utility is concave in its three arguments and satisfies the Inada conditions, such that the first-order conditions have interior solutions
\begin{eqnarray}
   && \frac{\partial U_t}{\partial n_t}=0\quad \Leftrightarrow \quad n_{t}=n(w_t,r_{t+1})>-1, \quad n_w=\frac{\partial n_t}{\partial w_t}>0 \quad n_r=\frac{\partial n_{t}}{\partial r_{t+1}}>0 \label{Fertilitynormal goods}\\
   && \frac{\partial U_t}{\partial s_t}=0 \quad \Leftrightarrow \quad s_t=s(w_t,r_{t+1})
\end{eqnarray}
Following \cite{Eas61}, we assume in (\ref{Fertilitynormal goods}) that children are ``normal goods.'' That is, $n_w>0$ and $n_r>0$ reflect that income effects increase fertility. Finally, with endogenous fertility, the life cycle savings condition is:
\begin{eqnarray} (1+n(f(k_t)-f'(k_t)k_t,f'(k_{t+1})))k_{t+1}=s(f(k_t)-f'(k_t)k_t,f'(k_{t+1})) \quad \Leftrightarrow \quad k_{t+1}=\psi(k_t) \label{17}\end{eqnarray}
such that differentiation around the steady state yields:
\begin{eqnarray} &&\psi_k=\frac{dk_{t+1}}{dk_t}=\frac{(n_wk-s_w)f''(k)k}{1+n+(n_rk-s_r)f''(k)}\label{18}\end{eqnarray}
Equation (\ref{18}) implies:
\begin{prop} \label{Prop 1 2024}
    Fertility evolves in cycles if and only if $\psi_k\in(-1,0)$. Fertility $n$ increases in the capital labor ratio $k$. \end{prop}
\begin{proof} Note that $k_{t+1}=\psi(k_t)$ such that $n_t(w_t(k_t),r_{t+1}(k_{t+1}))=n_t(w_t(k_t),r_{t+1}(\psi(k_{t})))=n_t(k_t)$. Hence, fertility moves in cycles with the capital intensity $k_t$. Moreover, near the steady state, we have $n_k:=\frac{dn}{dk}=\frac{\partial n}{\partial w}\frac{\partial w}{\partial k}+\frac{\partial n}{\partial r}\frac{\partial r}{\partial k}\frac{\partial k}{\partial k}=n_w(-f''(k)k)+n_rf''(k)\psi_k>0$, since $\psi_k<0$ and $n_w>0$ and $n_r>0$. That is, fertility $n$ is increasing in the capital labor ratio $k$ if children are normal goods.
\end{proof}

To evaluate these cycles and to answer the
question whether large or small cohorts are ``favored," we rewrite our model into the form (\ref{8})-(\ref{10}) that we used in the proof of Proposition \ref{c1}. This yields:\footnote{It is
straightforward that continuously differentiable functions
$\phi()$, $\phi^{-1}()$, $\psi()$, $\psi^{-1}()$, $\xi()$, and
$\xi^{-1}()$ exist as in Appendix \ref{a0}. Regarding $\phi()$, it
is important to note that
$\frac{dU_t}{dk_t}=\frac{1}{\phi_U}=\frac{\partial U_t}{\partial
k_t}+\frac{\partial U_t}{\partial k_{t+1}}\psi_k+\frac{\partial
U_t}{\partial n_t}\frac{\partial n_t}{\partial
k_t}=-U_{c^1}\Big((1+r)-\psi_k(1+n)\Big)\frac{1}{1+r}f''(k)k$,
where the derivative $\frac{\partial U_t}{\partial
n_t}\frac{\partial n_t}{\partial k_t}$ vanishes by the envelope
theorem, i.e., since $n_t$ is implied by the FOC $\frac{\partial
U_t}{\partial n_t}=0$.}
\begin{eqnarray} k_{t+1}=\psi(k_t), \quad \psi_k\in(-1,0) \label{19} \end{eqnarray}
\begin{eqnarray} k_t=\phi(U_t), \quad \frac{1}{\phi_U}=\frac{dU_t}{dk_t}=-U_{c^1}\Big((1+r)-\psi_k(1+n)\Big)\frac{1}{1+r}f''(k)k>0 \label{20}\end{eqnarray}
\begin{eqnarray} U_{t+1}=\xi(U_t), \quad \xi_u=\psi_k  \label{21}\end{eqnarray}
\begin{eqnarray} n_{t}=n(w(k_t),r(k_{t+1}))=n(w(k_t),r(\psi(k_t)))=n(k_t), \quad \frac{dn_t}{dk_t}>0 \label{22} \end{eqnarray}
Using (\ref{19})-(\ref{22}), we have

\begin{prop} Utility, fertility, and the capital labor ratio converge in cycles to the steady state, and the utility of small cohorts exceeds the utility of large cohorts. \label{p3}\end{prop}
\begin{proof} The first two statements follow immediately from (\ref{19})-(\ref{22}).
Regarding the favored cohort effect, we recall that $n_t$
represents generation $t's$ fertility and thus the relative (to
cohort $t$) size of cohort $t+1$. Hence we have to study
$\frac{dU_{t+1}}{dn_{t}}$, respectively $\frac{dn_t}{dU_{t+1}}$,
to obtain the condition under which small cohorts are favored.
Using (\ref{22}) yields
$n_t=_{|(\ref{22})}n(k_t)=_{|(\ref{20})}n(\phi(U_t))=_{|(\ref{21})}n(\phi(\xi^{-1}(U_{t+1})))$,
such that
$\frac{dn_{t}}{dU_{t+1}}=n_k\phi_u\frac{1}{\xi_u}=_{|(\ref{21})}n_k\phi_u\frac{1}{\psi_k}<0$ since $n_k>0$, $\phi_u>0$, while $\psi_k<0$.
\end{proof}

Proposition \ref{p3} shows that small cohorts are favored when children are normal goods. Interestingly, this condition is independent of the relative size of \(r\) and \(n\).

To see this, note that the covariance between fertility and capital implies that richer cohorts, born with a high wage rate \(w(k_t)\), have more children. In turn, once this large cohort enters the labor market, it reduces the capital-labor ratio and raises the interest rate \(r_{t+1}\). In other words, the small cohort \(t\) benefits from high wages, chooses high fertility, and enjoys high returns on savings. This effect is reversed for generation \(t+1\): they are more numerous, receive lower wages, have fewer children, and thus face low interest rates.

An alternative interpretation links our result to the classic golden-rule discussion. At a golden-rule steady state, a marginal change in capital has no first-order effect on steady-state utility. However, a small deviation $dk_0$ (for instance, due to a depreciation shock) triggers oscillatory dynamics. Importantly, the welfare consequences of these transitional fluctuations are independent of whether $r\gtrless n$.



\section{Related Literature}\label{rel}

The analysis of convergence in neoclassical overlapping generations models has traditionally been conducted in capital space, via the law of motion $k_{t+1}=\psi(k_t)$ as in \citet{Dia65}, \citet{Gal89}, and \citet{Cro02}. This paper proposes a complementary perspective by characterizing transitional dynamics directly in utility space through the first-order difference equation $U_{t+1}=\xi(U_t)$. Because this transformation preserves the convergence speed of the capital stock while directly tracking cohort lifetime utilities, it facilitates a transparent normative evaluation of non-steady-state paths.

Our framework is closely related to several strands of the literature. First, it builds upon and the \citet{Sam766} pioneering model of self-generated fertility waves. While Samuelson analyzed endogenous fertility cycles in a simplified OLG setup without capital accumulation or utility, the present paper embeds the Easterlin mechanism in a standard neoclassical OLG model with endogenous wages, interest rates, and capital. This allows us not only to generate fertility cycles but also to derive welfare implications for successive cohorts along the transition.

Second, the paper connects to the growing computational literature on endogenous fertility and demographic change. Studies such as \citet{LudVog10} and \citet{LudSchVog12} examine fertility, education, and welfare using large-scale numerical OLG models. Our analytical utility-space approach complements these quantitative exercises by offering a tractable theoretical benchmark for understanding the welfare consequences of fertility fluctuations without relying exclusively on simulations.

Compared to other analytically oriented research, the present framework serves as a convenient tool to study an economy's endogenous fertility patterns along non-steady state paths. Conversely, the methodology in \citet{Kuh14} remains more convenient to extend steady state models where demographic variations are induced parametrically, such as those investigated by \citet{Sam75}, \citet{Mic93}, \citet{Jae09}, \citet{Ponth2010}, \citet{cro12}, \citet{Pes14}, and \citet{Fel16}.

\section{Conclusion}\label{sg4}

This paper develops a tractable analytical framework for studying transitional dynamics in neoclassical OLG models by mapping the economy directly into utility space. We show that cohort lifetime utilities evolve according to a continuously differentiable first-order difference equation, $U_{t+1}=\xi(U_t)$. This reformulation offers a transparent and general method for normative analysis along non-steady-state paths without requiring explicit solutions to the underlying nonlinear system.

Applying this approach to endogenous fertility waves, we provide a characterization of the welfare consequences of Easterlin-type cycles. In particular, we show that when children are normal goods, small cohorts systematically enjoy higher lifetime utility than large cohorts. Moreover, we find that this cohort-welfare asymmetry is driven by fertility preferences rather than the economy's position relative to the golden rule.

The utility-space perspective complements both traditional capital-space analysis and large-scale computational OLG models. It offers a simple tool for evaluating intergenerational welfare. Future research could extend this framework to incorporate human capital accumulation, endogenous longevity, or quantitative calibrations that assess the magnitude of cohort utility differences.

\newpage

\newcounter{thesection}
\numberwithin{equation}{section}

\addcontentsline{toc}{section}{References}
\markboth{References}{References}
\bibliographystyle{apalike}
\bibliography{References}

\newpage

\begin{appendix}

\section{Existence of $\psi()$, $\phi()$, and $\xi()$}\label{a0}


We recall the implicit function theorem.\footnote{See
\citet{Lee03} pp. 164-166, for a proof of the implicit function
theorem and the inverse function theorem.}
\begin{prop} Let $x^*,y^*$ be point where the continuously differentiable function $F(x,y)=0$ is
satisfied. Moreover, let the partial derivative $F_y(x^*,y^*)\neq
0$. Then there exists a neighborhood $\mathcal{N}$ around
$x^*,y^*$ and a continuously differentiable function $g()$ such
that $y=g(x)$ and $\frac{dy}{dx}=g_x(x^*)=-\frac{F_x}{F_y}$ for
all $x,y\in\mathcal{N}$. If the derivative $F_x(x^*,y^*)$ is also
non-vanishing, there also exists an inverse function $g^{-1}()$
such that $x=g^{-1}(y)$, and
$\frac{dx}{dy}=g^{-1}_y(y^*)=-\frac{F_y}{F_x}$ for all
$x,y\in\mathcal{N}$.\label{imp}\end{prop}

To apply Proposition \ref{imp}, we note that it follows directly
from our assumptions that the primitive functions $U(,)$ and
$f()$, as well as the functions $s(,)$, $r()$, and $w()$ in
(\ref{8}) and (\ref{9}) are at least once continuously
differentiable.\footnote{In particular, the savings function
implied by the euler equation (\ref{4}) is once continuously
differentiable since $U(,)$ is by assumption twice continuously
differentiable.} Moreover, we know that the composition of
continuously differentiable functions is continuously
differentiable. Hence, to show the existence of $\phi()$,
$\phi^{-1}()$, $\psi()$, and $\psi^{-1}()$ by means of Proposition
\ref{imp}, it only remains to check that the derivatives
$\frac{dk_{t+1}}{dk_t}=\psi_k(k)$,
$\frac{dk_{t}}{dk_{t+1}}=\psi^{-1}_k(k)$,
$\frac{dk_t}{dU_t}=\phi_u(U)$, and
$\frac{dU_t}{dk_t}=\phi^{-1}_k(k)$, are well defined \emph{at} the
steady state, where $k_t=k_{t+1}=k$. With these functions in
place, the properties of $\xi()=\phi^{-1}(\psi(\phi()))$ in
(\ref{10}) and its inverse
$\xi^{-1}()=\phi^{-1}(\psi^{-1}(\phi()))$ follow at once.

\paragraph{Derivation of $\psi()$ and $\psi^{-1}()$:} To show that an implicit
function $\psi()$, that satisfies
\begin{eqnarray} &&(1+n)k_{t+1}=s(w(k_t),r(k_{t+1})) \quad \Leftrightarrow \quad k_{t+1}=\psi(k_t), \quad \psi\in
\mathcal{C}^1\label{81}\end{eqnarray} exists, we note that the
implicit derivative
$\frac{dk_{t+1}}{dk_t}=\psi_k(k)=\frac{-s_wf''(k)k}{(1+n)-s_rf''(k)}$
is well defined. This, follows immediately from assumption
(\ref{7}), which ensures that
$0<\frac{-s_wf''(k)k}{(1+n)-s_rf''(k)}<1$. Moreover, all functions in this fraction are at least continuous. 
Hence, according to Proposition \ref{imp}, there exists a function
$\psi()$, such that $k_{t+1}=\psi(k_{t})$. Finally, since
$-s_wf''(k)k\neq0$, the inverse function $\psi^{-1}()$ exists and
has the derivative
$\psi^{-1}_k(k)=\frac{(1+n)-s_rf''(k)}{-s_wf''(k)k}>0$.

\paragraph{Derivation of $\phi()$ and $\phi^{-1}()$:} Recalling the household
problem (\ref{3}) and (\ref{4}) we have
\begin{eqnarray}U_t&&=U\Big(w(k_t)-s(w(k_t),r(k_{t+1})),(1+r(k_{t+1}))s(w(k_t),r(k_{t+1}))\Big)=\tilde{U}(k_t,k_{t+1})\nonumber\\
&&=_{|(\ref{8})}\tilde{U}(k_t,\psi(k_t)) \quad \Leftrightarrow
\quad k_t=\phi(U_t), \quad \phi\in
\mathcal{C}^1.\label{91}\end{eqnarray}

To understand $\phi()$ in (\ref{91}), it is useful to start with
the derivative of its inverse:
\begin{eqnarray}\phi^{-1}_k(k)=\frac{dU_t}{dk_t}=\frac{\partial U}{\partial
k_t}+\frac{\partial U}{\partial
k_{t+1}}\psi_k(k).\label{92}\end{eqnarray} Using the Euler
equation (\ref{4}) as an envelope, we rewrite (\ref{92}) such that
\begin{eqnarray}\phi^{-1}_k(k)&&=\frac{dU_t}{dk_t}=U_{c^1}\Big(-f''(k)k+\psi_k(k)f''(k)\frac{s}{1+r}\Big)\label{93}\\
&&=_{|(\ref{81})}-U_{c^1}\Big(k-\psi_k(k)\frac{(1+n)k}{1+r}\Big)f''(k),\nonumber\\
&&=-U_{c^1}\Big((1+r)-\psi_k(k)(1+n)\Big)\frac{1}{1+r}f''(k)k.
\label{94}\end{eqnarray} Hence, $\phi^{-1}_k(k)$ is continuous in
the capital intensity. 

However, (\ref{93}) may be zero. In this case, we have
$\frac{dU_t}{dk_t}=0$, and, \emph{locally}, our first-order
approximation does not record any changes in utility that
originate from changes in the capital intensity. 
The \emph{local} dynamics of utility are trivial in the sense that
utility is constant.\footnote{\label{f2} We may note that
(\ref{93}) implies that this special case, in which there are
levels of $k$ at which utility does not change (locally) with
capital, is only possible for the Walrasian economy which is dynamically inefficient $(r<n)$ such
that $0<\frac{1+r}{1+n}=\psi_k(k)<1$. In Marshallian economies,
which are cyclically stable, we have $-1<\psi_k(k)<0$ and utility
always changes with capital, i.e., derivative (\ref{93}) is never
zero. 

See \citet{Dia65}, p.
1132 and p. 1148 for the case where $\psi_k\in(-1,0)$. This case
can only occur if the life-cycle savings condition (\ref{6})
reaches equilibrium through \emph{changes in the quantity} of
capital/savings supplied (Marshallian adjustment). If the capital
market reaches equilibrium through \emph{changes in price
(interest)} (Walrasian adjustment), capital must converge
monotonically. \citet{Sam41}, pp. 102-106, and \citet{Hic1939},
pp. 115-129, discuss the Walrasian and the Marshallian market
mechanism in detail.} In all cases of interest, where utility changes over time,
$\phi()$ is well-defined and we have:
\begin{eqnarray}   \phi_u(U)=\frac{dk_t}{dU_t}=\frac{1}{\Big(-U_{c^1}\Big((1+r)-\psi_k(k)(1+n)\Big)\frac{1}{1+r}f''(k)k\Big)}.\label{95}\end{eqnarray}

\paragraph{Derivation of $\xi()$ and $\xi^{-1}()$:} Taken together, paragraphs 1 and 2 show that, as long as
utility changes over time, the function
$\xi()=\phi^{-1}(\psi(\phi()))$ in (\ref{10}) is well-defined. In
particular since $\psi()$, $\psi^{-1}()$, $\phi()$, and
$\phi^{-1}()$ are all continuously differentiable functions,
$\xi()$ and $\xi^{-1}()=\phi^{-1}(\psi^{-1}(\phi()))$ are also
both continuously differentiable functions. Hence, $\xi()$ is a
$\mathcal{C}^1$ diffeomorphism, which is what we needed to show. 

\end{appendix}

\end{document}